\documentclass[copyright,creativecommons]{eptcs}
\providecommand{\event}{}
\providecommand{\volume}{??}
\providecommand{\anno}{20??}
\providecommand{\firstpage}{1}
\providecommand{\eid}{}

\usepackage{subfigure}
\usepackage{breakurl}
\usepackage{color,graphics,xcolor}
\usepackage{graphicx}
\usepackage{amsmath}
\usepackage{amssymb,amsthm}
\usepackage{proof}
\usepackage{graphicx}
\usepackage[all]{xy}
\usepackage{amssymb}
\usepackage{fancybox}
\usepackage{type1cm}
\usepackage{url}
\usepackage{listings}
\usepackage{xytree}

\title{Heuristic Methods for Security Protocols\thanks{The authors thank the anonymous reviewers for their valuable suggestions and Alberto Lluch-Lafuente for the useful discussions.}}
\author{Qurat ul Ain Nizamani
  \institute{Department of Computer Science \\
    University of Leicester, UK}
  \email{qn4@mcs.le.ac.uk}
  \and
  Emilio Tuosto
  \institute{Department of Computer Science \\
    University of Leicester, UK}
  \email{et52@mcs.le.ac.uk}
}
\def\titlerunning{Heuristic Methods for Security Protocols}
\def\authorrunning{Qurat ul Ain Nizamani \& Emilio Tuosto}

\input{prooftree}
\newcommand{\pnf}{\mathit{pnf}}
\newcommand{\pl}{$\mathcal{PL}$}
\newcommand{\conf}[1]{\langle {#1} \rangle}
\newcommand{\enc}[2]{\{#1\}_{#2}}
\newcommand{\sep}{\;\;\mid\;\;}
\newcommand{\comment}[1]{}
    
\newcommand{\tool}{$\mathcal{A}${\sc SPAS}{\tt y}{\sc A}}

\begin{document}
\maketitle
\begin{abstract}
  Model checking is an automatic verification technique to verify
  hardware and software systems. However it suffers from state-space
  explosion problem. In this paper we address this problem in the context
  of cryptographic protocols by proposing a \emph{security
    property-dependent} heuristic. The heuristic weights the state
  space by exploiting the security formulae; the weights may then be
  used to explore the state space when searching for attacks.
\end{abstract}
\section{Introduction}\label{intro:sec}
Security protocols present many interesting challenges from both
pragmatic and theoretical points of view as they are ubiquitous and
pose many theoretical challenges despite their apparent simplicity.
One of the most interesting aspects of security protocols is the
complexity of the verification algorithms to check their correctness;
in fact, under many models of the intruder, correctness is undecidable
and/or computationally hard~\cite{dlms99,ckmrt03,ckmrt03-b,shm04}.

Many authors have formalised security protocols in terms of process
calculi suitable to define many verification frameworks (besides model
checking, path analysis, static analysis,
etc.)~\cite{low96,ag99,bod01,bor01}.
Model checking (MC) techniques have been exploited in the design and
implementation of automated tools~\cite{cjm98,mms97} and
\emph{symbolic} techniques have been proposed to tackle the state
explosion problem~\cite{alv03,bb05,bbn04,bft06}.

This paper promotes the use of directed MC in security protocols.
We define a heuristic based on the logic formulae formalising the
security properties of interest and we show how such heuristic may
drive the search of an attack path.
Specifically, we represent the behaviour of a security protocol in the
context of the (symbolic) MC framework based on the cIP and \pl, respectively
a cryptographic process calculus and a logic for specifying security
properties introduced in~\cite{bft06}.
An original aspect of this framework is that it allows to explicitly
represent instances\footnote{ This key feature of the framework is
supported by the use of \emph{open variables}, a linguistic
mechanism to tune and combine instances that join the run of a
protocol. } of participants and predicate over them.

Intuitively, the heuristic ranks the nodes and the edges of the state
space by inspecting (the syntactical structure of the) formula
expressing the security property of interest.
More precisely, the state space consists of the transition system
representing possible runs of a protocol; our heuristic weights states
and transitions considering the instances of principals that joined
the context and how they are quantified in the security formula.
Weights are designed so that most promising paths are tried before
other less promising directions.
Rather interestingly, the heuristic can rule out a portion of the
state space by exploring only a part of it.
In fact, we also show that the heuristic may possibly cut some
directions as they cannot lead to attacks; in fact, the heuristic is
proved to be correct, namely no attacks can be found in the portion
of the state space cut by our heuristic.

At the best of our knowledge, heuristic methods have not yet been
explored to analyse security protocols (at least in the terms proposed
in this paper) which may be surprising.
In general, many of the features of heuristics fit rather well with the
verification of security protocols.
More precisely, $(i)$ optimality of the solution (``the attack'') is
not required when validating protocols (violations of properties of
interest are typically considered equally harmful), $(ii)$ the
graph-like structures (e.g., labeled transition systems)
representing the behaviour of security protocols usually have
'symmetric regions' which may be ignored once one of them is
checked, $(iii)$ heuristic search may be easily combined with
several verification frameworks and particularly with MC.
The lack of such research thread is possibly due to the fact that it
is in general hard to define heuristics for security protocols.
In fact, typically heuristics are tailored on (properties of) a
\emph{goal state} and measure the ``distance'' from a state to a
goal state.
In the case of the verification of security protocols, this would boil
down to measure the distance from an attack where a security property
is violated.
Therefore, designing heuristics suitable to improve MC of security
protocols is hard as attacks cannot be characterised beforehand.

Here, we argue that heuristic search may be uniformly used in the
verification of security protocols and provide some interesting cases
of how our approach improves efficiency.
In fact, we will illustrate how the use of the heuristic can greatly
improve the efficiency of the search by cutting the directions that
cannot contain attacks.

Our approach seems to be rather promising, albeit this research is in
an initial stage, it can be extended in many directions.
Finally, we argue that our proposal can be applied to other
verification frameworks like~\cite{bor01} or inductive proof methods
like~\cite{pau97,thg99} (see \S~\ref{rw}).

\noindent
\textbf{Structure of the paper.}  \S~\ref{bg} summarizes the concepts
necessary to understand our work; \S~\ref{ph} yields the definition of
our heuristic which is then evaluated and proved correct in
\S~\ref{ha:sec}; \S~\ref{rw} concludes the paper and discusses related work.

\section{Background}\label{bg}

\newtheorem{thm}{Theorem}[section]
\newtheorem{defn}{Definition}
\newtheorem{lemm}[thm]{Lemma}

This section fixes our notation (\S~\ref{cip:sec}) and a few basic concepts
on informed search largely borrowed from~\cite{k8} (\S~\ref{h:sec})

\subsection{Expressing security protocols and properties in cIP and \pl}\label{cip:sec}
We adopt the formal framework introduced in~\cite{bft06} consisting of
the cIP (after \emph{cryptographic Interaction Pattern}) process
calculus and the \pl\ logic (after \emph{protocol logic}) to respectively
represent security protocols and properties.
Here, we only review the main ingredients of cIP and \pl\ by means of
the Needham-Shroeder (NS) public key protocol and refer the reader
to~\cite{bft06} for a precise presentation.

The NS protocol consists of the following steps
\[\begin{array}{lll}
1. & A \to B: & \enc{na, A}{B^+} \\
2. & B \to A: & \enc{na,nb}{A^+} \\
3. & A \to B: & \enc{nb}{B^+}
\end{array}\]
where, in step $1$ the initiator $A$ sends to $B$ a nonce $na$ and her
identity encrypted with $B$'s public key $B^+$; in step 2, $B$
responds to the nonce challenge by sending to $A$ a fresh nonce $nb$
and $na$ encrypted with $A^+$, the public key of $A$; $A$ concludes
the protocol by sending back to $B$ the nonce $nb$ encrypted with
$B$'s public key.

In cIP principals consist of their identity, the list of \emph{open
  variables} and the actions they have to perform in the protocol.
A cIP principal can either send or receive messages from a public channel
using the $out$ and $in$ actions respectively.
The NS protocol can be formalized in cIP as follows:
\begin{equation}\label{ns:eq}
  \begin{array}{c@{\hspace{2cm}}c}
    \begin{array}{ll}
      A: (r)[
        & out(\{na, A\}_{r^+}).\\
        & in(\{na,?z\}_{A^-}). \\
        & out(\{z\}_{r^+})
        \ \ \ ]
    \end{array}
    &
    \begin{array}{ll}
      B: ()[
        & in(\{?x,?y\}_{B^-}).\\
        & out(\{x,nb\}_{y^+}).\\
        & in(\{nb\}_{B^-})
        \ \ \ ]
    \end{array}
  \end{array}
\end{equation}
The principal $A$ (resp. $B$) in~(\ref{ns:eq}) represents the
initiator (resp. the responder) of the NS protocol.
The open variable $r$ is meant to be bound to the identity of the
responder.
The principal $A$ first executes the output action and then waits
for a message which should match the pattern specified in the $in$
action.
More precisely, $A$ will receive any pair encrypted with her public
key whose first component is the nonce $na$; upon a successful match,
the second component of the pair will be assigned to the variable $z$.
For instance, the $\enc{na,M}{A^+}$ matches $\enc{na,?z}{A^-}$ for any
$M$ and would assign $M$ to $z$.

\bigskip

We adopt the definition of \pl\ formulae given in~\cite{bft06}:
\[\begin{array}{lcl}
  \phi,\psi & ::= &
  x_i = m \sep
  \kappa \rhd m \sep
  Qi:A. \psi \sep
  \lnot \psi \sep
  \psi \land \phi \sep
  \psi \lor \phi
\end{array}\]
where $Q$ ranges over the set of quantifiers $\{\forall,\exists\}$
and $x_i$ are indexed variables (a formula without
quantifiers is called \emph{quantifier-free}).

The atomic formulae $x_i = m$ and $\kappa \rhd m$ hold respectively
when the variable $x_i$ is assigned the message $m$ and when
$\kappa$ (representing the intruder's knowledge) can derive $m$.
Notice that quantification is over indexes $i$ because \pl\ predicates
over the \emph{instances} of the principals concurrently executed.
A principal \emph{instance} is a cIP principal indexed with a
natural number; for example, the instance of the NS initiator
obtained by indexing the principal $A$ in~(\ref{ns:eq}) with $2$ is
\begin{equation}\label{instance:eq}
 A_{2}:(r_{2})[out(\{na_{2}, A_{2}\}_{r_{2}^{+}}).in(\{na_{2},?z_{2}\}_{A_{2}^{-}}).out(\{z_{2}\}_{r_{2}^{+}})]
 \end{equation}
 we let $[X]$ be the set of all instances of a principal $X$; e.g.,
 the instance\footnote{ Hereafter, we denote an instance simply by the
   indexed name of the principal; for example, the instance above will
   be referred to as $A_2$.  } in~(\ref{instance:eq}) is in $[A]$
(\S~\ref{ph} illustrates how the transition system of cIP instances
of a protocol is obtained).

As an example of \pl\ formula consider the formula $\psi_{NS}$
predicating on (instances of) the NS protocol:
\[
  \forall i: A.\ \exists j:B\ (x_{j} = na_{i} \ \wedge \ z_{i} = nb_{j}).
\]
The formula $\psi_{NS}$ states that for all instances of $A$ there
should be an instance of $B$ that has received the nonce $na_i$ sent
by $A_i$ and the nonce $nb_j$ is received by the instance $A_i$.

\subsection{Basics of heuristics}\label{h:sec}
As mentioned in \S~\ref{intro:sec}, the approaches such as symbolic MC
can be used to tackle the problem of state space explosion.
However, even with the use of such approaches the search space can
grow enormously.
It is therefore desirable to look into methods through which search
space can be generated/explored more efficiently using \emph{informed
  search algorithms} that are characterized by the use of a
\emph{heuristic function} (also called evaluation function).
A heuristic function assigns a weight to nodes by estimating their
``distance'' from a \emph{goal node}.

We recall here the basic concepts on heuristic algorithms by
means of a simple example and refer the reader to~\cite{k8} for a
deeper presentation.

The $n$-puzzle (also known as the sliding-block or tile-puzzle) is a
well-known puzzle in which the goal is to move square tiles by sliding
them horizontally or vertically in one empty tile.
For $n=8$ the goal configuration is depicted in Figure~\ref{fig:goal};
a possible initial configuration is in Figure~\ref{fig:start}.
The problem of finding the shortest path leading to the goal
configuration is NP-hard.
\begin{figure}[h]
\begin{minipage}[b]{0.5\linewidth}
\centering
\includegraphics[width=3.18cm,height=2.67cm]{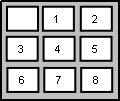}
\caption{The 8-puzzle goal configuration} \label{fig:goal}
\end{minipage}
\hspace{0.5cm}
\begin{minipage}[b]{0.5\linewidth}
\centering
\includegraphics[width=3.18cm,height=2.67cm]{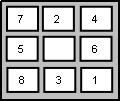}
\caption{A possible start configuration} \label{fig:start}
\end{minipage}
\end{figure}


A very simple heuristic (cf.~\cite{k8}) for 8-puzzle can be given by
\[
h_{1} = \text{number of misplaced tiles.}
\]
For each configuration, $h_{1}$ counts the number of misplaced tiles
with respect to the goal configuration.
For instance, $h_1$ weights 8 the configuration in Figure~\ref{fig:start}
since all the tiles are misplaced.

Another heuristic (cf.~\cite{k8}) for 8-puzzle is the one that
exploits the so called \emph{Manhattan distance}
\[
h_{2} = \text{sum of the Manhattan distances of non-empty tiles from their target positions.}
\]
So the configuration in Figure~\ref{fig:start} is weighted 18 by
$h_{2}$.

An important property of heuristics is \emph{admissibility}; an
admissible heuristic is one that never over estimates the cost to
reach the goal node.
Both $h_{1}$ and $h_{2}$ are admissible; in fact, $h_1$ is clearly
admissible as each misplaced tile will require at least one step to
be on its right place, and $h_2$ is also admissible as at each step
the tiles will be at most one step closer to goal.
A non-admissible heuristic is $h_{3}=h_1 * 4$; in fact, if only one
tile is misplaced with respect to a goal configuration, $h_{3}$ will
return 4 which is an overestimation of the distance to goal.

\section{A  Heuristic for Security Protocols}\label{ph}
This section introduces our original contribution (\S~\ref{heu}), namely the
heuristic for effectively searching the state space generated for MC cryptographic protocols.
As mentioned earlier, there is not much work done in this regard.
Specifically (to the  best of our knowledge) no work exists that can
prune a state space in  verification of cryptographic protocols.

The heuristic function is defined on the security formula expressed in
\pl.

The heuristic is efficient as it will not only guide the searching
algorithm towards promising regions of the graph but can also prune
those parts of the state space where attack cannot happen under a
given security formula.

The heuristic is defined in terms of two mutually recursive functions
$\mathcal{H}_{s}$ and $\mathcal{H}_{t}$ which assign weights to states
and transitions respectively.
The state space is obtained according to the semantics of cIP defined
in \cite{bft06}.
For lack of space, an informal presentation of the semantics is given
here.

\subsection{The state space}
A state consists of a tuple $\langle\mathcal{C},\chi,\kappa\rangle$
where
\begin{itemize}
\item $\mathcal{C}$ is a context containing principal instances which
  joined the session,
\item $\chi$ is a mapping of variables to messages, and
\item  $\kappa$ is a set of messages representing the intruder knowledge.
\end{itemize}
A transition from one state to another can be the result of \emph{out}
and \emph{in} actions performed by principal instances or of
\emph{join} operations non-deterministically performed by the intruder;
join transitions may instantiate open variables by assigning them with
the identity of some principal (provided it is in $\kappa$).
Initially, $\mathcal C$ is empty and therefore the only possible transitions
are join ones.
When $\mathcal C$ contains an instance ready to send a message, an \emph{out}
transition can be fired so that the sent message is added to $\kappa$.
If $\mathcal C$ contains a principal ready to receive a message, the
intruder tries to derive from the messages in $\kappa$ a message that
matches the pattern specified in the input action (see \S~\ref{bg}); if
such a message is found $\chi$ is updated to record the assignments to
the variables occurring in the input action.

For instance, a few possible transitions for the NS protocol are
\[
  s_0 \stackrel{join}\longrightarrow s_1 \stackrel{join}\longrightarrow s_2\stackrel{out}\longrightarrow s_3 \stackrel{in}\longrightarrow s_4
\]
where $s_0 = \langle \ \emptyset, \emptyset, \kappa_0  \ \rangle$
with $\kappa_0 = \{ I,I^{+},I^{-} \}$, namely initially no principal instance
joined the context, there is no assignment to variables, and the intruder
only knows its identity and public/private keys.

The join transition from $s_0$ to $s_1$ adds a principal instance
$B_2$ to the context yielding
\begin{align*}
  s_{1}=  \langle \  \{&()[ in(\{?x_{2},?y_{2} \}_{B_{2}^{-}}).out( \{ x_{2},nb_{2}\}_ {y_{2}^{+}}).in(\{nb_{2} \}_{B_{2}^{-}})] \},\\& \emptyset, \\& \kappa_1=\kappa_0 \cup \{B_{2},{B_{2}}^+ \} \ \rangle
\end{align*}
that is, the intruder now knows $B_2$'s identity and (by default) its
public key.
Similarly, the transition from $s_1$ to $s_2$ adds the principal
instance $A_1$ to context and therefore
\begin{align*}
  s_2=\langle \ \{&()[ in(\{?x_{2},?y_{2} \}_{B_{2}^{-}}). out( \{ x_{2},nb_{2} \}_{y_{2}^{+}}).in(\{nb_{2} \}_{B_{2}^{-}})], ()[out(\{na_{1},
  A_{1}\}_{B_{2}^{+}}).in(\{na_{1},?z_{1}\}_{A_{1}^{-}}).out(\{z_{1}\}_{B_{2}^{+}})] \}, \\ & \{r_{1}\mapsto B_{2}\}, \\ & \kappa_2 =\kappa_1 \cup \{A_{1},
  {A_{1}}^+ \}\ \rangle
\end{align*}
Notice that the open variable $r_1$ is now mapped to $B_{2}$.

The transitions from $s_2$ to $s_3$ is due to an $out$ action
\begin{align*}
  s_3=\langle \ \{&()[ in(\enc{?x_{2},?y_{2}}{B_{2}^{-}}).out( \{
  x_{2},nb_{2} \}_{y_{2}^{+}}).in(\{nb_{2}\}_{B_{2}^{-}})],
  ()[in(\{na_{1},?z_{1}\}_{A_{1}^{-}}).out(\{z_{1}\}_{B_{2}^{+}})]\} ,
  \\& \{r_{1}\mapsto B_{2}\},\\& \kappa_3= \kappa_2 \cup \{
  na_{1},A_{1}\}_{B_{2}^{+}} \ \rangle
\end{align*}
the prefix of $A_1$ is consumed and the message is added to the
intruder's knowledge.

Finally, the transition from $s_3$ to $s_4$ is due to an \emph{in}
transition for the input prefix of $B_2$.
The message $\enc{na_{1},A_{1}}{B_{2}^{+}}$ added to the intruder's
knowledge in the previous transition matches the pattern
$\enc{?x_{2},?y_{2}}{B_{2}^{-}}$ specified by $B_2$, therefore the
$x_{2}$ and $y_{2}$ are assigned to $na_1$ and $A_1$ respectively.
Hence,
\begin{align*}
  s_4=\langle \ \{& out( \{ na_1,nb_{2} \}_{A_{1}^{+}}).in(\{nb_{2}\}_{B_{2}^{-}})],()[in(\{na_{1},?z_{1}\}_{A_{1}^{-}}).out(\{z_{1}\}_{B_{2}^{+}})]\},
  \\& \{r_{1}\mapsto B_{2} ,x_{2} \mapsto na_1, y_{2} \mapsto A_1\},
  \\ &  \kappa_3  \ \rangle
\end{align*}


In our framework, join transitions can be safely anticipated before
any other transition (Observation 10.1.3 in~\cite{tuo02}, page 174).

\subsection{The heuristic}\label{heu}
For simplicity and without loss of generality, we define the
heuristic on \emph{Prenex Normal Form} (PNF) formulae defined below.
\begin{defn} \label{pnf1} [Prenex Normal Form]
  A \pl\ formula is in \emph{prenex normal form} if it is of the form
  \[Q_{1}i_{1}:A_{1}. \cdots .Q_{n}i_{n}:A_{n}.\phi\]
  where $\phi$ is a quantifier-free formula and, for $1 \leq j \leq
  n$, $Q_j \in \{ \forall, \exists\}$, each $i_j$ is an index
  variable, and $A_j$ is a principal name.
\end{defn}
Basically, a PNF formula is a formula where all the quantifiers are
"at top level".
Notice that, in Definition~\ref{pnf1}, it can be $n=0$ which amounts
to say that a quantifier free formula is already in PNF.
\begin{thm}\label{pnf}
Any \pl\ formula can be transformed into a logically equivalent PNF
formula.
\end{thm}
\begin{proof}
Let the function $\pnf: $\pl $\rightarrow $\pl\ be defined as follows:
\[
  \pnf(\psi)= \begin{cases}
   \psi  & \psi \ \text{is a quantifier-free formula.}
   \\
   Qi:A.\pnf(\psi') &\psi \equiv Qi:A.\psi'
   \\
   \overline {Q}i:A.\pnf(\neg\psi'')
   &\psi \equiv\neg\psi' \text{ and } \pnf(\psi') \equiv Qi:A.\psi''\\
   Qi':A.\pnf(\psi_{1}'[i'/i] \wedge \psi_{2})
   & i' \text{ fresh}, \psi \equiv \psi_{1}
   \wedge \psi_{2} \text{ and }
   \pnf(\psi_{1})\equiv Qi:A.\psi_{1}' \\
   Qi':A. \pnf(\psi_{1}'[i'/i] \vee \psi_{2})
   & i' \text{ fresh}, \psi \equiv \psi_{1}
   \vee \psi_{2} \text{ and } \pnf(\psi_{1}) \equiv Qi:A.\psi_{1}'
 \end{cases}
\]
The proof of theorem \ref{pnf} follows from the properties of $\pnf$
given by Lemmas \ref{transform} and \ref{logical} below.
\end{proof}
\begin{lemm}\label{transform}
For any \pl formula $\psi$, $\pnf(\psi)$ is in PNF.
\end{lemm}
\begin{proof}
  We proceed by induction on the structure of $\psi$.

  If $\psi$ is a quantifier free formula then it is in PNF and,
  by definition of $\pnf$, $\pnf(\psi) = \psi$.

  The inductive case is proved by case analysis.
  \begin{itemize}
  \item Assume $\psi$ is $Qi:A.\psi'$, then by definition of $\pnf$,
    $\pnf(\psi)= Qi:A.\pnf(\psi') $. By inductive hypothesis
    $\pnf(\psi')$ is in PNF and therefore $\pnf(\psi)$ is in PNF.
  \item If $\psi= \psi_{1} \land \psi_{2}$ then, assuming $
    \pnf(\psi_{1})= Qi:A.\psi_{1}'$, by definition of $\pnf$
    \[ \pnf(\psi) = Qi':A. \pnf ( \psi_{1}'[i'/i]\wedge\psi_{2}) \]
    For fresh index $i'$ not occuring in $\psi_2$.
    By inductive hypothesis $ \pnf(\psi_{1}'[i'/i]\land\psi_{2})$ is in
    PNF, therefore $ \pnf(\psi)$ is in PNF.
  \item The case $\psi= \psi_{1} \lor \psi_{2}$ is analogous.
  \item If $\psi= \lnot\psi'$ then, assuming $
    \pnf(\psi')=Qi:A.\psi''$, by definition of $\pnf$, $ \pnf(\psi) =
    \overline {Qi}:A. \pnf(\lnot\psi'')$. By inductive hypothesis
    $\pnf(\lnot\psi'')$ is in PNF, therefore $\pnf(\psi)$ is in PNF.
  \end{itemize}
\end{proof}

\begin{lemm}\label{logical}
$ \pnf(\psi)\Leftrightarrow \psi$.
\end{lemm}
\begin{proof}
  We proceed by induction on the structure of $\psi$.

  If $\psi$ is a quantifier free formula then $\pnf (\psi) = \psi$ and
  therefore $ \pnf (\psi) \Leftrightarrow\psi$.

  Again the proof for the inductive case is given by case analysis.
  \begin{itemize}
  \item Assume $\psi$ is $Qi:A.\psi'$, then by definition of $\pnf$, $\pnf(\psi)= Qi:A.\pnf(\psi') $. By inductive hypothesis
    $\pnf(\psi') \Leftrightarrow \psi'$ hence $\pnf(\psi)\equiv
    Qi:A. \psi'$ and therefore $\pnf(\psi) \Leftrightarrow \psi$.
  \item If $\psi=\psi_{1}\wedge \psi_{2}$ then, assuming $\pnf
    (\psi_{1})= Qi:A.\psi_{1}'$, by definition of $\pnf$
    \[ \pnf(\psi)= Qi':A.\pnf(\psi_{1}'[i'/i]\land\psi_{2}) \] For $i'$
    fresh (namely, $i'$ does not occur in $\psi_{2}$).
    By inductive hypothesis $ \pnf (\psi_{1})\Leftrightarrow\psi_{1}$
    hence
    \[ \psi\Leftrightarrow \pnf (\psi_{1})\land
    \psi_{2}=(Qi:A.\psi_{1}')\land \psi_{2} \]
    It is trivial to prove that for any \pl\ formula $(Qi: A.\psi) \land
    \phi \Leftrightarrow Qi:A.(\psi \land \phi)$ and therefore
    $\pnf(\psi)\Leftrightarrow\psi$.
  \item The proof for $\psi= \psi_{1} \lor \psi_{2}$ is similar.
  \item If $\psi= \lnot \psi'$ then, assuming $ \pnf (\psi')=
    Qi:A.\psi''$, by definition of $\pnf$, $\pnf
    (\psi)=\overline{Qi}:A. \pnf (\lnot\psi'')$. By inductive
    hypothesis $ \pnf (\psi')\Leftrightarrow \psi'$ and therefore
    $\psi\Leftrightarrow\lnot pnf(\psi')$, hence $\psi \Leftrightarrow
    \lnot(Qi:A.\psi'') \Leftrightarrow \overline {Qi}:A.\lnot\psi''$
    and therefore $ \pnf (\psi) \Leftrightarrow\psi$.
  \end{itemize}
\end{proof}

The heuristic function $\mathcal{H}_{s}$ is given in
Definition~\ref{hc} and depends on the function $\mathcal H_t$ given
in Definition~\ref{ha} below.
\begin{defn}[Weighting states]\label{hc}
  Given a state $s$ and a formula $\phi$, the \emph{state weighting
    function} is given by
  \[
    \mathcal{H}_{s}(s, \phi) =
    \begin{cases}
      \understackrel{\max}{t\in s\Upsilon} \ {\mathcal{H}_{t}(t,\phi)}, & s\Upsilon \ \neq\ \emptyset
      \\
      -\infty, & \phi \equiv \forall i:A. \ \phi' \ \wedge \
       s \Upsilon = \emptyset \ \wedge \ s \cap [A]=\emptyset
      \\
      0, & \text{otherwise}
    \end{cases}
  \]
  where $s\Upsilon$ is the set of join transitions departing from $s$
  and, assuming $s = \langle \mathcal C,\chi,\kappa\rangle$, $s \cap
  [A]$ stands for $\kappa \cap [A]$.
\end{defn}
The function $\mathcal{H}_{s}$ takes a state, say $s = \langle
\mathcal{C},\chi,\kappa\rangle$, and a formula $\phi$ as input and
returns the maximum among the weights computed by $\mathcal H_t$ on
the join transitions departing from $s$ for $\phi$.
The weight $-\infty$ is returned if
\begin{itemize}
\item $\phi$ is a universal quantification on a principal instance $A$
  ($\forall i:A. \ \phi'$),
\item $s$ does not have outgoing join transitions
  ($s\Upsilon = \emptyset)$, and
\item there is no instance of $A$ in the context
  ($s \cap [A]=\emptyset$).
\end{itemize}

The heuristic $\mathcal{H}_{s}$ has been designed considering that a
formula universally quantified on instances of $A$ is falsified in
those states where there is at least one instance of $A$.
Therefore a context that does not have an instance of the quantified
principal, has no chance of falsifying the formula.
In fact, the condition $s\Upsilon = \emptyset$ ensures that no
principal instance can later join the context.
As a result, there is no possibility of falsifying the property in all
paths emerging from this state which can therefore be pruned.
This justifies the second case of $\mathcal{H}_s$ where the value
$-\infty$ is assigned to such states.

The heuristic that assigns weights to transitions is given in
Definition~\ref{ha}.
\begin{defn}[Weighting transitions]\label{ha}
  Given a state $s$ and a transition $t$ from $s$ to
  $s'= \langle\mathcal{C}',\chi',\kappa'\rangle$
  in $s\Upsilon$, the \emph{weighting transitions} function
  $\mathcal{H}_{t}$ is
  \[
  \mathcal{H}_{t}(t,\phi) = \begin{cases}
    1+\mathcal{H}_{s}(s',\phi'),
    & \phi\equiv\ \forall i:A.\ \phi'\ \wedge\ \kappa'\cap[A] \neq \emptyset ,\\
    1+\mathcal{H}_{s}(s',\phi),
    & \phi\equiv \exists i: A.\ \phi'\ \wedge \ \kappa' \cap[A] = \emptyset ,\\
    \mathcal{H}_{s}(s',\phi'),
    &\phi\equiv \exists i:A.\ \phi' \ \wedge \ \kappa' \cap [A] \neq \emptyset,\\
    \mathcal{H}_{s}(s',\phi),
    &\phi\equiv\ \forall i:A.\ \phi' \ \wedge \ \kappa' \cap [A]=\emptyset, \\
    0 & otherwise.
  \end{cases}
\]
\end{defn}
The function $\mathcal{H}_t$ takes as input a transition $t$ and
invokes $\mathcal H_s$ to compute the weight of $t$ depending on the
structure of the formula $\phi$.
As specified in Definition~\ref{ha}, the value of the weight of the
arrival state is incremented if either of the two following mutually
exclusive conditions hold:
\begin{itemize}
\item $\phi$ universally quantifies on a principal instance $A$
  ($\forall i:A.\ \phi'$) for which some instances have already joined
  the context ($\kappa'\cap[A] \neq \emptyset$);
\item $\phi$ existentially quantifies on a principal instance $A$
  ($\exists i:A.\ \phi'$) which is not present in the context
  ($\kappa' \cap [A] = \emptyset$).
\end{itemize}
Instead, the heuristic $\mathcal H_t$ does not increment the weight
of the arrival state if either of the following mutually
exclusive\footnote{
  Note that all the conditions of the definition of $\mathcal H_t$ are
  mutually exclusive.
} conditions hold:
\begin{itemize}
\item $\phi$ existentially quantifies on instances of $A$ ($\exists
  i:A.\ \phi'$) and present in the context ($\kappa'\cap[A] \neq \emptyset$);
\item $\phi$ universally quantifies on instances of $A$ ($\forall
  i:A.\ \phi'$) and the context does not contain such instances
  ($\kappa'\cap[A] = \emptyset$).
\end{itemize}
Again the intuition behind $\mathcal{H}_{t}$ is based on quantifiers.
The formula $\phi$ that universally (resp. existentially)
quantifies on instances of $A$ can be falsified only if such instances
will (resp. not) be added to the context.
Therefore all transitions that (resp. do not) add an instance of $A$
get a higher value.
 It is important to mention that in the first
and third cases of Definition~\ref{ha}, the recursive call to
$\mathcal{H}_s$ takes in input $\phi'$, the subformula of $\phi$ in
the scope of the quantifier. This is due to the fact that once an
instance of the quantified principal has been added we are not
interested in more instances and therefore consume the quantifier.
The heuristic $\mathcal H_t$ returns $0$ when $\phi$ is a quantifier
free formula.
In fact, due to the absence of quantifiers we cannot assess how
promising is $t$ to find an attack for $\phi$.
We are investigating if in this case a better heuristic is possible.

Finally, we remark that $\mathcal H_s$ and $\mathcal H_t$ terminate
on a finite state space because the sub-graph consisting of the join
transitions forms a tree by construction\footnote{
  For page limits we do not prove it formally, but it can easily be
  checked by the informal description of join transitions given in
  this section.
}.
Therefore, the recursive invocations from $\mathcal H_t$ to
$\mathcal H_s$ will eventually be resolved by the last two cases of
$\mathcal H_s$ in Definition~\ref{hc}.

 \section {Evaluation of the Heuristic}\label{ha:sec}
In this section we describe with the help of examples how $\mathcal H_s$ and $\mathcal H_t$ can find attacks without exploring the complete state
space.
In the first example the heuristic is applied on the NS protocol and
in the second example it is applied on the KSL protocol.

We also prove the correctness of the heuristic.

\subsection{Applying the heuristic to the Needham-Schroeder protocol}\label{ns:sec}
Let us consider the property $\psi_{NS}$ given in \S~\ref{cip:sec} as
$\forall i: A.\ \exists j:B\ (x_{j} = na_{i} \ \wedge \ z_{i} =
nb_{j})$.
Figure~\ref{figure1} illustrates a portion of the state space of the NS protocol
after the first two join transitions when $\psi_{NS}$ is considered.
Notice that $\psi_{NS}$ can be falsified in a path where there is a
context containing at least one instance of A and no instances of B.
\begin{figure}[h]
  \subfigure[Join transitions of the NS protocol]{
    \centering
    \includegraphics[scale=.5]{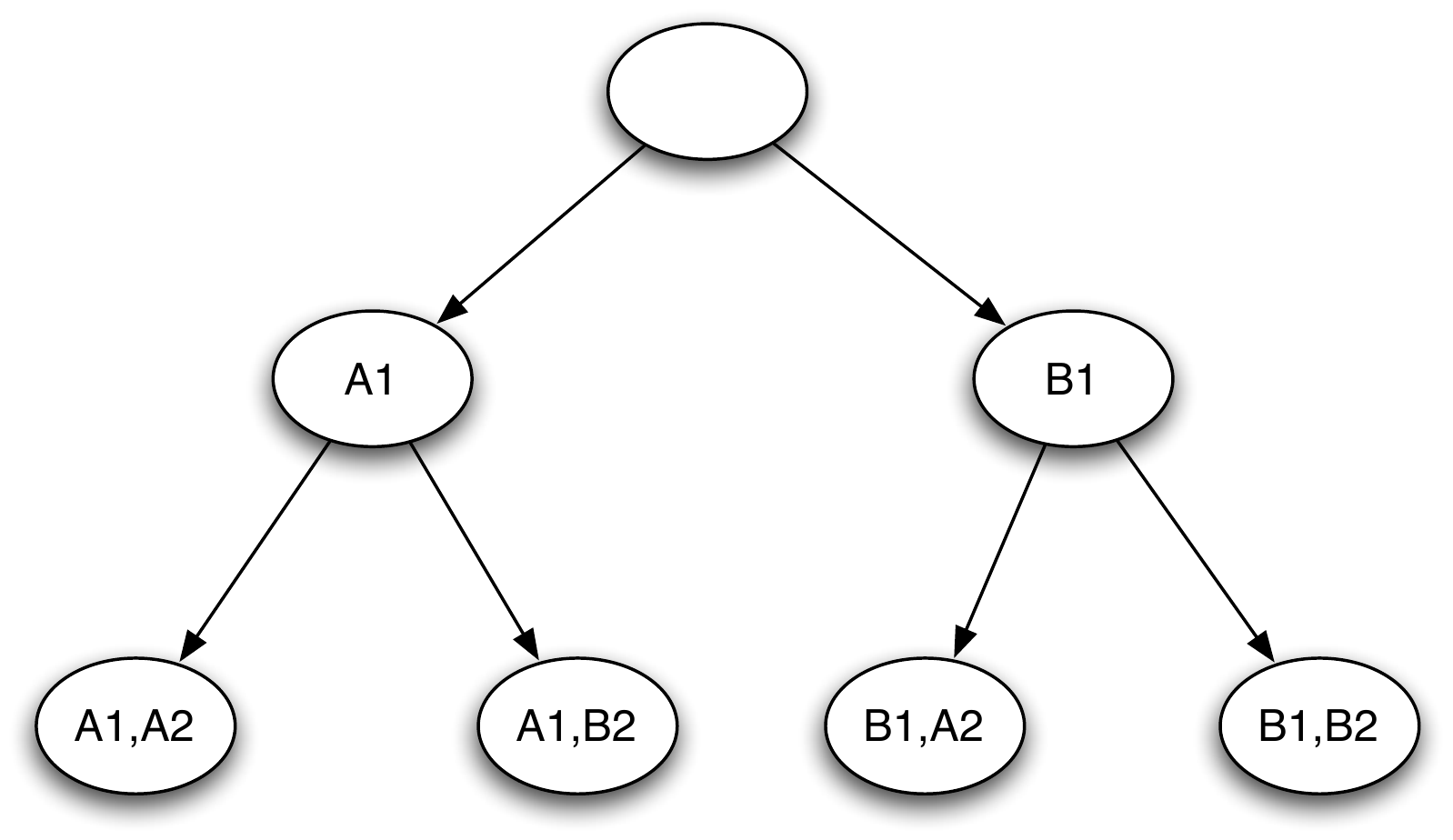}
    \label{figure1}
  }
  \subfigure[Weighted states in the NS protocol]{
    \centering
    \includegraphics[scale=.5]{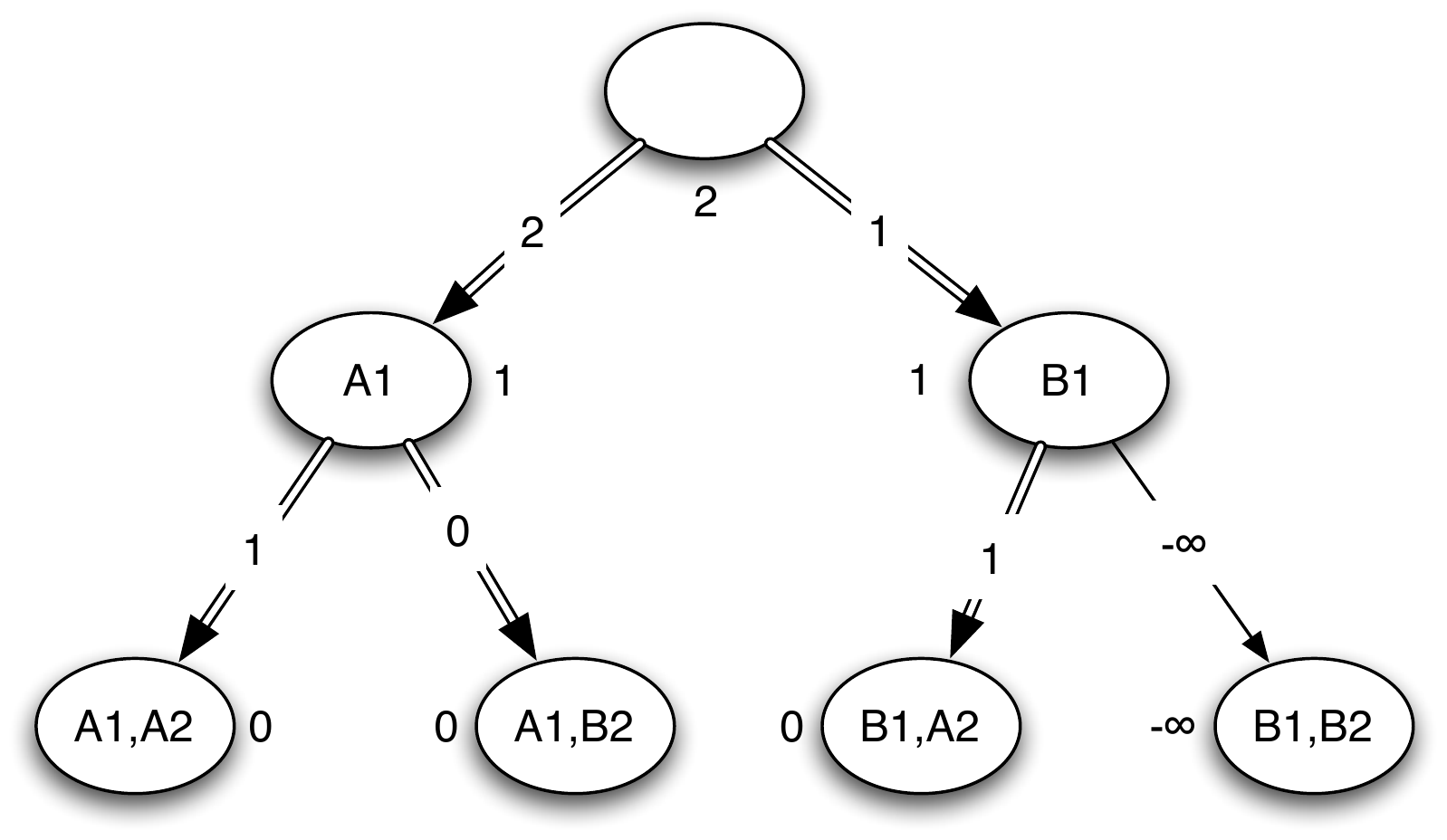}
    \label{figure2}
  }
  \caption{Join transitions and weighted states in NS protocol}
\end{figure}

The heuristic will assign weights to states and transitions as in
Figure~\ref{figure2}. The highlighted paths (those with 'fat' arrows)
are the one to be explored; the context $\{A_1,A_2\}$ contains the attack
reported below:
\[\begin{array}{lll@{\hspace{1cm}}l}
1. & A_1 \to I : & {\{na_1,A_1\}}_{I^+}  \\
2. & A_2 \to I : & {\{na_2,A_2\}}_{I^+}  \\
3. & I \to A_1 : & {\{na_1,na_2\}}_{{A_1}^{+}}  & \kappa \rhd na_1,na_2 \\
4. & I \to A_2 : & {\{na_2,na_1\}}_{{A_2}^+} & \kappa \rhd na_1,na_2 \\
5. & A_1 \to I : & {\{na_2\}}_{I^+}\\
6. & A_2 \to I : & {\{na_1\}}_{I^+}
\end{array}\]
The intruder acts as responder for both $A_1$ and $A_2$. As a result of
step 1 and 2, $\kappa$ contains $na_1$ and $na_2$; enabling the
intruder to send messages to $A_1$ and $A_2$ at step 3 and 4
respectively. This results into assignments like $z_{A_1}= na_2$ and
$z_{A_2}= na_1$, which is the falsification of stated property which
requires a nonce generated by an instance of $B$ to be assigned to the
variables.

The other two highlighted paths contain a similar attack, we
report the one with context $\{A_1,B_2\}$.
\[\begin{array}{llll}
A_1 \to I : & {\{na_1,A_1\}}_{I^+}   \\
I \to A_1 : & {\{na_1,I\}}_{{A_1}^{+}}, &\kappa \rhd na_1,I \\
A_1 \to I : & {\{I\}}_{I^+}\\
I \to B_2 : & {\{na_1,I\}}_{{B_2}^+},  &\kappa \rhd na_1,I \\
B_2 \to I : & {\{na_1,nb_2\}}_{I^+}  \\
I \to B_2 : & {\{nb_2\}}_{{B_2}^+},  &\kappa \rhd nb_2
\end{array}\]
Again at step 2 and 4, $A_1$ and $B_2$ are receiving the identity of
intruder instead of nonce by $B$, resulting into an attack.

It is evident from the Figure~\ref{figure2} that heuristic assigns
appropriate weights to the paths that contain an attack. It is
worthy mentioning that the context $\{ B_1,B_2\}$ has been labeled
$-\infty$, therefore the search will never explore this state. This
suggests that approximately 1/4th of the state space can be pruned
by applying heuristic.  This is a rough estimate taking into
consideration the symmetry in the state space (the context $\{A_{1},
A_{2}\}$ is similar to $\{B_{1}, B_{2}\}$ and $\{A_{1}, B_{2}\}$ is
similar to $\{B_{1}, A_{2}\}$).

\subsection{Applying the heuristic to the KSL protocol}\label{ksl:sec}
We consider the analysis of (the second phase of) KSL \cite{k15}, done in
\cite{bft06}.
The protocol provides repeated authentication and has two phases; in
the first phase $(i)$ a trusted server $S$ generates a session key
$kab$ to be shared between $A$ and $B$, and $(ii)$ $B$ generates the
\emph{ticket} $\enc{Tb, A, kab}{kbb}$ for $A$ (where $Tb$ is a
timestamp and $kbb$ is known only to $B$).

In the second phase, $A$ uses the ticket (until it is valid) to
repeatedly authenticate herself to $B$ without the help of $S$.
The second phase can be specified as follows:
\[\begin{array}{lll}
 1. & A \to B: & na, \enc{Tb, A, kab}{kbb}\\
 2. & B \to A: & nb, \enc{na}{kab}\\
 3. & A \to B: & \enc{nb}{kab}
\end{array}\]
$A$ sends a fresh nonce $na$ and the ticket to $B$ that accepts the
nonce challenge and sends $nb$ together with the cryptogram
$\enc{na}{kab}$ to $A$.
In the last message, $A$ confirms to $B$ that she got $kab$.

In $cIP$, $A$ and $B$ can be represented as follows:
\[
  \begin{array}{c@{\hspace{2cm}}c}
    \begin{array}{ll}
      A: (b, sk, tk)[
    & out(na,\{b,A,sk\}_{tk}).\\
    & in(?y,\{na\}_{sk}).\\
    & out(\{y\}_{sk})] \\
    \end{array}
    &
    \begin{array}{ll}
      B: (a, sk, tk) [
    & in(?x,\{B,a,sk\}_{tk}).\\
    & out(nb,\{x\}_{sk}).\\
    & in(\{nb\}_{sk})]\\
    \end{array}
   \end{array}
\]
(where for simplicity the timestamp generated by $B$ is substituted by his
identity). Authentication is based on the mutually
exchanged nonces, and formalized as follows:
\[
\psi_{KSL}= \forall l: B.\ \forall j:A.(b_j = B_l \wedge a_l = A_j \to x_l = na_j \wedge y_j = nb_l)
\]
which reads any pair of properly connected ``partners'' $B_l$ and
$A_j$ ($b_j = B_l \wedge a_l = A_j$) eventually exchange the nonces
$na_j$ and $nb_l$.

\begin{figure}[h]
  \subfigure[2 principals]{
    \centering
    \includegraphics[scale=.4]{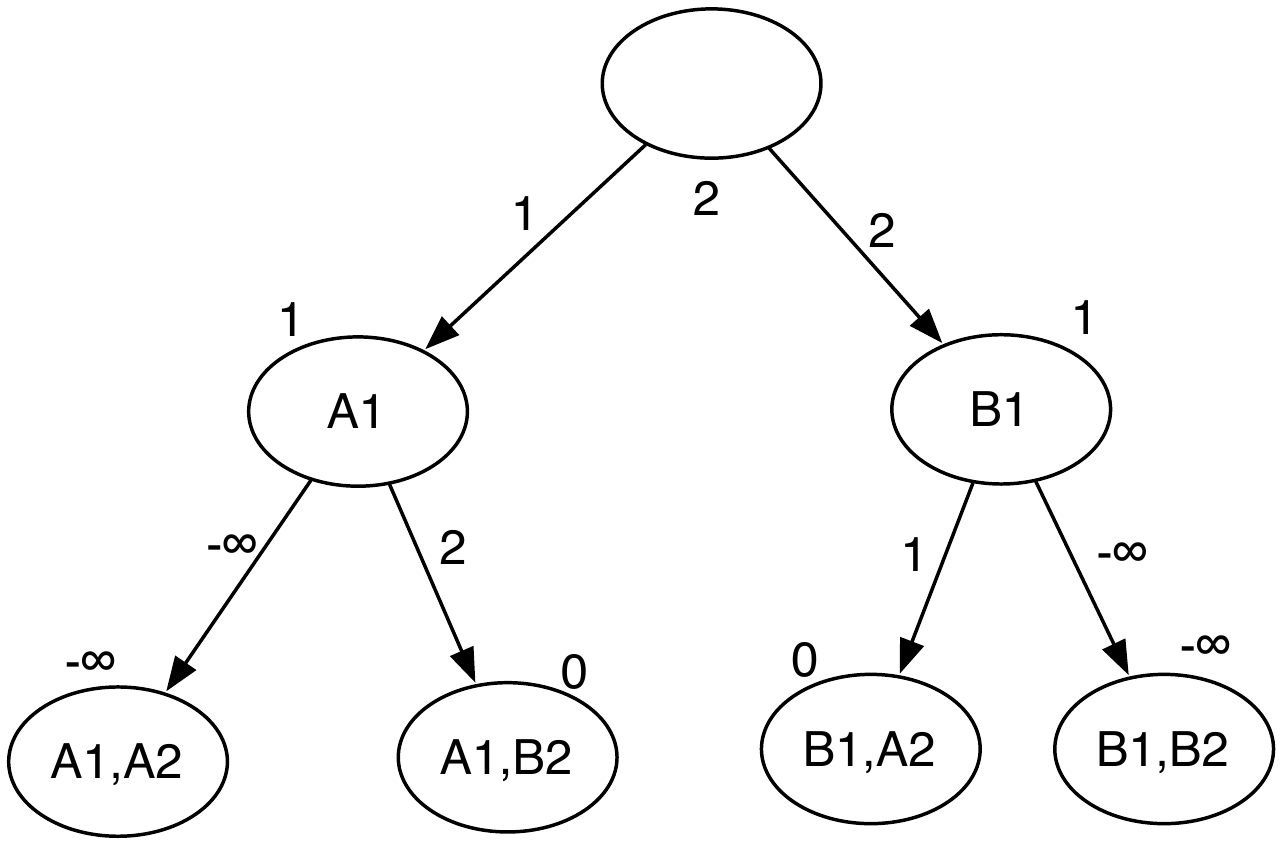}
    \label{ksl2}
  }
  \subfigure[3 principals]{
    \centering
    \includegraphics[scale=.4]{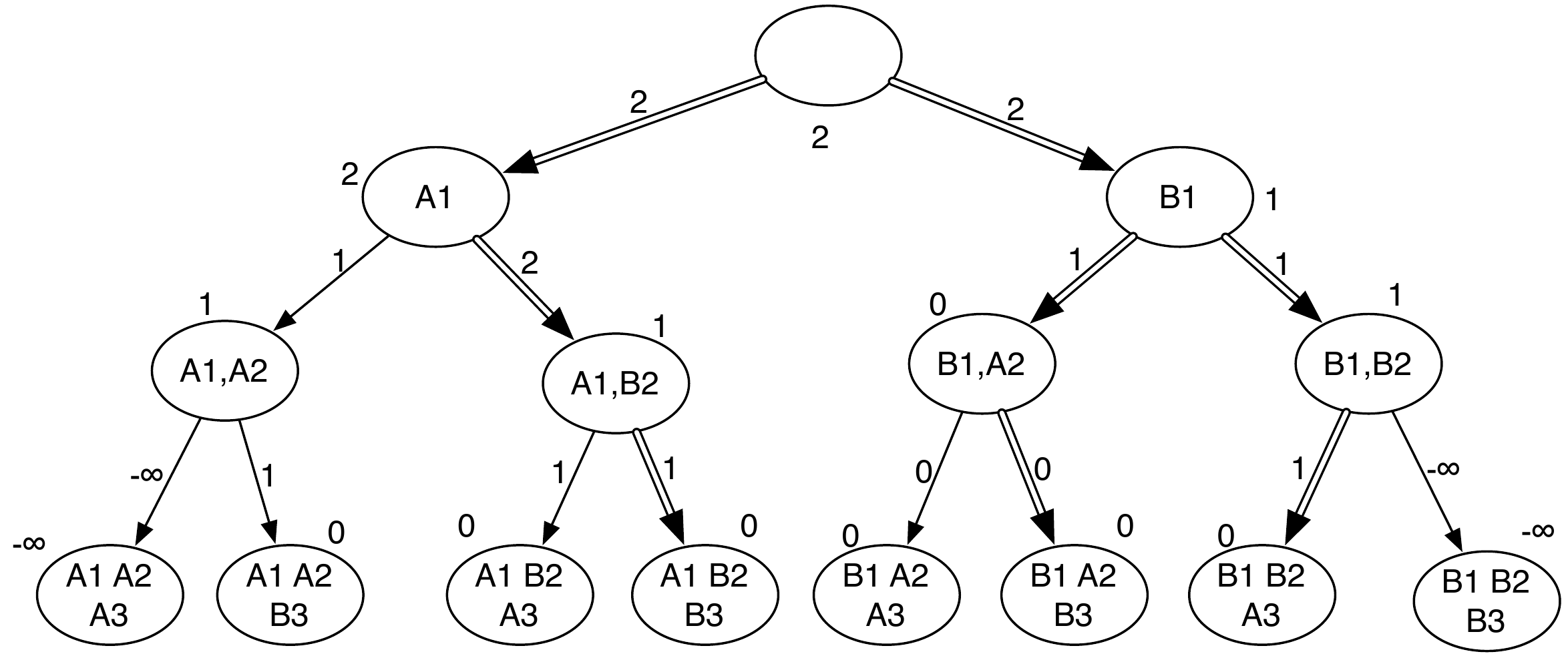}
    \label{ksl}
  }
  \caption{Join transitions of KSL}
\end{figure}
Figures~\ref{ksl2} and~\ref{ksl} depict the weighted join
transitions for 2 and 3 principal instances respectively.
The verification with 2 principal instances reports no attack and
the conclusion can be derived by just exploring half of the state
space (the context $\{A_1,A_2\}$ and $\{B_1,B_2\}$ are labeled
$-\infty$; see Figure~\ref{ksl2}). In case of 3 principal instances
the attacks are found in highlighted paths (those with 'fat' arrow
in Figure~\ref{ksl}). The heuristic assigns appropriate weights to
such paths and 2 states are labeled $-\infty$, suggesting a rough
cut down of 1/4th of the state space.

The examples show that heuristic is able to guide the searching
algorithm towards promising paths containing attacks. Moreover a
considerable part of the state space is pruned, reducing the number
of states to be explored by searching algorithm.

\subsection{Properties of  $\mathcal{H}_{s}$ and $\mathcal{H}_{t}$}\label{properties}
First we would like to briefly comment on the admissibility of our
proposed heuristic.
Admissibility of heuristics is important in certain problems where it
is possible to reach many goal states along different paths each path
having a different cost.
Hence, it may be not only important to find a goal state, but also
find the goal state on the path with the best (or an acceptable)
cost (as discussed in \S~\ref{h:sec} for the $n$-puzzle).
In such cases, it is important for heuristic function to return an
estimation of the cost to reach a goal from the state.

We contend that for security protocols the situation is different.
In fact, the goal state in this case is an ``attack'', namely a state
that violates the security property.
Typically, it is very hard to compare the importance of different
attacks as the violation of a property may be due to many causes
as for the NS example in \S~\ref{ns:sec}).
Therefore, optimality of the attack is of less concern when
validating protocols; what matters in the first instance is to find
an attack, if any. However, we envisage the problem of finding
optimal solutions as important but we do not consider it in this
paper.

It is also important to remark that the weights assigned by
$\mathcal H_s$ and $\mathcal H_t$ to states or transitions do not
correspond to evaluate the proximity to a target state.
Rather they estimate the likeliness for the state to lead to an
attack. This leads to a different scenario where the heuristic
function does not have to return the cost to reach at goal node.
Rather our heuristic returns a value that corresponds to the chance
that nodes and transitions are on a path leading to an attack.
We therefore contend that admissibility is not an issue in our
case.


The following theorem proves the correctness of our heuristic; namely,
it shows that pruned parts of the state space do not contain any
attack.
\begin{thm}
If $\mathcal{H}_{s}(s,\phi)=
-\infty$ then for any state s' =
$\langle\mathcal{C'},\chi',\kappa' \rangle$ reachable from
s= $\langle\mathcal{C},\chi,\kappa \rangle,
\kappa'\models_{\chi'} \phi$
\end{thm}
\begin{proof}(Sketch)
  Suppose that there is a $s'$ reachable from $s$ such that
  $\kappa'\not\models_{\chi'}\phi$. Then $\kappa'\models_{\chi'} \neg
  \phi$ and by Definition of $\models$ (the relation $\models$ is
  reported in Appendix~\ref{model:app}), there is $A_n \in \kappa'$(
  because $\phi\equiv \forall i: A. \ \phi'$). However by hypothesis
  $s\cap[A]= \emptyset$ and $s\Upsilon= \emptyset$ hence
  $s'\cap[A]=\emptyset$ and therefore $s'$ does not satisfy $\neg
  \phi$.
\end{proof}

\section{Concluding Remarks}\label{rw}
We have designed a heuristic that can be applied to improve the MC of
cryptographic protocols.
The proposed heuristic can drive the searching algorithm towards
states containing attacks with respect to a security formula.
Our heuristic may possibly prune parts of the state space that do not
contain attack.
We have shown that the heuristic is correct, namely we showed that
pruned parts of the state space do not contain attacks.

The formal context to define the heuristic is the one proposed
in~\cite{bft06} which features the cIP calculus and an ad-hoc logical
formalism, called \pl, to respectively express protocols and security
properties.
An original aspect of \pl\ is that it can quantify over principal
instances.
Formulas of \pl\ are checked against the (symbolic) semantics of cIP
by a tool called \tool\ (Automatic Security Protocol Analysis via a
SYmbolic model checking Approach) \cite{k4}.

\subsection{Related work}
At the best of our knowledge, the use of heuristics to analyse
cryptographic protocols has not been much studied.

The concept of heuristics in cryptographic protocol verification has
been utilized in~\cite{k13}.  The idea is to construct a
pattern\footnote{E.g. a pattern can be a representative of all traces
  violating secrecy}, $pt=(E,\rightarrow)$, where $E$ is a set of events
and $\rightarrow$ is a relation on the events.  Afterwards, it is
checked if a $pt$ can give realizable patterns which are actual traces
of the protocol and represent an attack.  For each event execution,
there are certain terms that need to be in intruder's knowledge or
that are added to intruder's knowledge represented by $in(e)$ and
$out(e)$ respectively.  A process called pattern refinement is applied
to get realizable patterns for those events whose $in(e)$ requirements
are not satisfied.  An open goal represents such requirements and is
selected from set of potential open goals on the basis of the
heuristics.  In~\cite{k13}, 5 heuristics have been reported (e.g., an
open goal is selected randomly, open goals that require a decryption
key have higher priority).  However, the whole state space must be
searched if there are no attacks.  We argue that our approach can give
better results as it can prune certain parts of state space even when
there are no attacks (as seen in \S~\ref{ksl:sec}).

In~\cite{k10}, heuristics have been used to minimize the branching
factor for infinite state MC of security protocols.
Mainly, the heuristics in~\cite{k10} reorder the nodes, for instance
actions involving intruder are rated higher than actions initiated by
honest participants.
However these heuristics are very basic and as noted in~\cite{k11},
the tool does not scale to most of the protocols.

Though heuristic methods have not received much attention for MC security protocols, they have been studied for MC in general. 

In ~\cite{k3}, a heuristic has been defined in terms of model and
formula to be verified, that can also prune the state space.
Our heuristic seems to fall under the general conditions considered
in~\cite{k3} and we plan a deeper comparison.

In~\cite{k1}, the heuristic namely 'NEXT' compresses a sequence of
transitions into a single meta transition.
This eliminates transient states and therefore searching algorithm
does less work to find the goal node.
Similarly, in~\cite{k2} heuristics for safety and liveness for
communication protocols are given. At the best of our knowledge the
heuristics in \cite{k2} (and references therein) allow to cut the
state space only in few trivial cases.

\subsection{Future work}
This paper proposes the first step of a research program that may develop
in several directions.

First, other heuristics can be designed and studied; in fact, we are
planning to define two heuristics.
The former exploits the intruder's knowledge $\kappa$ and the cIP
protocol specification while the other heuristic exploits
\emph{joining formulae}, another feature (also supported by \tool) of
cIP.
Joining formulae are \pl\ formulae which enable the analyst to express conditions
on how principals should be joined (by predicating over open
variables)\footnote{
For example, runs of the NS protocol with at least an initiator and a
responder can be specified by the formula $(\exists i:A. true) \wedge
(\exists j:B. true)$ which rules out states that contain only instances
of $A$ or of $B$.
}.

The first heuristic will rank states considering the actions that principal
instances are ready to execute with respect to the formula to falsify.
For instance, if the goal is to prove that a variable should not be assigned
a given value, the heuristic may rank higher those states that assigns such
variable.

The second heuristic may instead be used to avoid the anticipation of
all the joining formulae at the beginning (which may be
computationally expensive) and use them to decide which instance to
introduce in a given state.

It will also be rather interesting to study the combined effect of
those heuristics (e.g., to consider their sum, or the max, etc.) or
also use multiple heuristics during the search depending on the
structure of the state. For instance, in one state one heuristic
might be more suitable than others. Further we intend to implement
these heuristics into existing tool in order to determine the
efficiency achieved in terms of space and time.

We also plan to consider heuristics in other verification contexts.
For instance, using \emph{strand spaces}~\cite{thg99}, the approaches
in~\cite{ath01,sri01} express properties in terms of connections
between strands.
A strand can be parameterised with variables and a trace is generated
by finding a substitution for which an interaction graph exists.
These approaches provide devices very similar to the join mechanism of
cIP and possibly be suitable for heuristics similar to ours to help in
finding the solution of the constraints.
Also, in~\cite{bb05} a symbolic semantics based on unification has
been adopted to verify security protocols with correspondence assertions
and the use of trace analysis.
We think that also in this case heuristics may drive the search for an
attack in a more efficient way.

\bibliographystyle{acm}
\bibliography{stringdef,emi,all,biblio}

\appendix
\appendix
\section{Model for \pl\ formulae}\label{model:app}
We borrow from~\cite{bft06} the definition of models for \pl.
\begin{defn}[Model for \pl\ formulae\label{definition.models}]
\index{$\kappa \models_\chi \phi$} Let $\chi$ be a mapping from
indexed variables to indexed messages, $\kappa$ a knowledge and $\phi$
a closed formula of \pl.
Then $\conf{\kappa,\chi}$ {\em is a model for $\phi$\/} if
$\kappa \models_\chi \phi$ can be proved by the
following rules (where $n$ stands for an instance index):

{\small
$$
\begin{array}{ccc}
\prooftree
   xa_n\chi = m\chi
\justifies
\kappa  \models_\chi xa_n = m
\using
(=)
\endprooftree
&\
&
\prooftree
 \kappa \rhd m\chi
\justifies
\kappa \models_\chi  \kappa \rhd m
\using
(\rhd)
\endprooftree
\end{array}
$$}
{\small
$$
\begin{array}{c}
\prooftree
   \mbox{ exists } n\mbox{ s.t. }  A_n \in \kappa
   \quad
   \kappa \models_\chi \phi [n/i]
\justifies
\kappa \models_\chi \exists i:A.\ \phi
\using (\exists)
\endprooftree
\\
\
\\
\prooftree
   \mbox{ forall } n\mbox{ s.t. }  A_n \in \kappa
   \quad
   \kappa \models_\chi \phi [n/i]
\justifies
\kappa \models_\chi \forall i:A.\ \phi
\using (\forall)
\endprooftree
\end{array}
$$}
{\small
$$
\begin{array}{ccccccc}
\prooftree
\kappa \models_\chi \phi \hspace{.3cm} \kappa \models_\chi \psi
\justifies
\kappa \models_\chi \phi \land \psi
\using (\land)
\endprooftree
&\
&
\prooftree
\kappa \models_\chi \phi
\justifies
\kappa \models_\chi \phi \lor \psi
\using (\lor1)
\endprooftree
\\
\
\\
\prooftree
\kappa \models_\chi \psi
\justifies
\kappa \models_\chi \phi \lor \psi
\using (\lor2)
\endprooftree
&\
&
\prooftree
\kappa \not\models_\chi \phi
\justifies
\kappa \models_\chi \lnot\phi
\using (\lnot)
\endprooftree
\end{array}
$$}
\end{defn}

\end{document}